\documentclass[conference]{IEEEtran}
\IEEEoverridecommandlockouts

\usepackage{cite}
\usepackage{amsmath,amssymb,amsfonts}
\usepackage{algorithmic}
\usepackage{graphicx}
\usepackage{grffile}
\usepackage{siunitx}
\usepackage{stfloats}
\usepackage{epstopdf}
\usepackage[caption=false,font=normalsize,labelfont=sf,textfont=sf]{subfig}
\usepackage{enumitem}
\usepackage{amsthm}
\newtheorem{theorem}{Theorem}
\newtheorem{lemma}{Lemma}

\usepackage{textcomp}
\usepackage{xcolor}

\def\BibTeX{{\rm B\kern-.05em{\sc i\kern-.025em b}\kern-.08em
    T\kern-.1667em\lower.7ex\hbox{E}\kern-.125emX}}
\usepackage[top=0.74in, bottom=1in, left=0.625in, right=0.625in]{geometry}

\begin{document}

\title{CSI Feedback Under Basis Mismatch: Rate-Splitting Transform Coding for FDD Massive MIMO}

\author{\IEEEauthorblockN{Youngmok Park, Bumsu Park and Namyoon Lee}
\IEEEauthorblockA{\textit{Department of Electrical Engineering,} \\
POSTECH,\\
Pohang, South Korea \\
Email: \{ympark1999, bumsupark, nylee\}@postech.ac.kr}
}

\maketitle

\begin{abstract}
In frequency division duplex massive multiple-input multiple-output systems, downlink  channel state information must be fed back within a limited uplink budget. While transform coding with Karhunen-Lo\`eve transform and reverse water-filling is rate-distortion optimal for Gaussian channels, its performance is limited by basis mismatch between the user and base station. We analyze this mismatch and propose a practical architecture separating long-term basis feedback from short-term coefficient quantization. Using a random vector quantization, we derive a closed-form end-to-end mean square error expression. This allows us to characterize the optimal rate split and identify a phase transition threshold for basis updates. Simulations on correlated Gaussian and COST2100 channels demonstrate near-optimal performance, robustness to update overhead, and significant complexity reduction compared to deep-learning-based autoencoders.
\end{abstract}

\section{Introduction}
Massive multiple-input multiple-output (MIMO) has become a foundational technology for achieving high spectral efficiency, spatial multiplexing, and wide-area coverage in modern wireless systems\cite{marzetta2010noncooperative,5GMag,2014LarssonMassive}. As networks evolve toward 6G, antenna arrays at the base station (BS) are expected to scale even further, giving rise to extremely large-scale MIMO architectures with orders of magnitude more spatial degrees of freedom\cite{6GMag}. While such scaling promises unprecedented performance gains, it also makes multi-user beamforming increasingly sensitive to inaccuracies in the channel state information (CSI) available at the BS\cite{limitedfeedback2008,Jindal2006,Han2024,Kim2025}. This sensitivity creates a critical bottleneck in frequency-division duplex (FDD) systems, where the downlink CSI must be compressed and fed back from the user equipment (UE) over a limited-rate control channel\cite{Lee2014,Caire2006}.

 To reduce this feedback overhead, two main classes of approaches have been widely investigated: compressive sensing (CS) and deep learning (DL)\cite{CSbasedCSI2017,CNN2020,deeplearning2018,Attention_CSI_Feedback_2022,Lightweight_CSI_feedback_2021,Guo2022_Overview,Learningtopitmize_2025,park2024multirate,park2025transformer}. CS-based methods exploit angular-domain sparsity\cite{CSbasedCSI2017,CNN2020}, but rely on strict assumptions and iterative reconstruction, which can degrade performance and incur high latency in rich-scattering environments. Recently, DL-based CSI feedback has emerged as an alternative\cite{deeplearning2018,Attention_CSI_Feedback_2022,Lightweight_CSI_feedback_2021,Guo2022_Overview,Learningtopitmize_2025,park2024multirate,park2025transformer}, utilizing autoencoder architectures to learn non-linear channel structures. While effective in reducing distortion, such approaches often incur prohibitive UE-side computational and memory costs and lack theoretical interpretability of the rate distortion (R-D) trade-off.

In this paper, we address these limitations by revisiting the CSI feedback problem from a fundamental source coding perspective. While practical massive MIMO channels exhibit structured and geometry-dependent characteristics, they are often approximated as correlated complex Gaussian vectors when focusing on second-order statistics. This approximation is physically motivated by the superposition of many independent scattering components at each antenna, for which the complex central limit theorem implies approximately Gaussian small-scale fading. As the number of antennas grows, non-Gaussian effects average out, making the covariance matrix a sufficient descriptor of the underlying propagation geometry. This correlated Gaussian abstraction is commonly adopted in standardized and measurement-based channel models, including the 3GPP SCM \cite{3GPP38901}, COST2100 \cite{COST2100}, and the NYU mmWave measurements \cite{NYUSIM}. For such sources, information theory dictates that the R-D optimum is achieved by transform coding (TC), using the Karhunen-Lo\`eve transform (KLT) and reverse water-filling (RWF) \cite{cover2006elements}.

A critical challenge, however, arises in the practical implementation of TC for FDD systems: the basis mismatch. Unlike classical source coding where the encoder and decoder share the optimal basis, the UE estimates the eigenbasis from local measurements and the BS lacks access to this precise geometry in FDD systems, where reciprocity-based estimation is inherently unavailable. Consequently, the basis itself must be quantized and fed back, consuming a portion of the limited rate budget\cite{limitedfeedback2008}. This introduces a fundamental trade-off between the bits allocated for describing the channel geometry (the basis) and those for the instantaneous fading (the scalar coefficients).

We formulate this as a joint R-D optimization problem under an encoder-decoder mismatch constraint. Our main contributions are: (i) a closed-form end-to-end (E2E) mean square error (MSE) expression that accounts for both coefficient quantization and Grassmannian basis quantization, (ii) an optimal bit split between basis and coefficients in the high-rate regime, and (iii) a phase-transition result that identifies the minimum feedback rate at which allocating bits to basis feedback becomes beneficial.

\section{System Model and Problem Formulation}
We consider a single-cell FDD downlink with a BS having $N_t$ antennas and a single-antenna UE. An OFDM waveform with $N_c$ subcarriers is used. The channel can be represented by a matrix $\mathbf{H} = [\mathbf{h}_1, \cdots, \mathbf{h}_{N_c}] \in \mathbb{C}^{N_t \times N_c}$, where $\mathbf{h}_k\in \mathbb{C}^{N_t}$ denotes the spatial channel at the $k$-th subcarrier. We define the stacked spatial-frequency channel vector as
\begin{equation}
    \bar{\mathbf{h}} = [\mathbf{h}_1^\top, \dots, \mathbf{h}_{N_c}^\top]^\top \in \mathbb{C}^{N},
\end{equation}
where $N = N_c N_t$ represents the total complex dimension of the channel.

\subsection{Channel Statistics and Geometry} 
Massive MIMO channels exhibit structured statistics governed by propagation geometry. Consistent with measurement campaigns and standardized models such as COST2100 and 3GPP SCM \cite{3GPP38901,COST2100}, the downlink channel is commonly modeled as a correlated complex Gaussian random vector:
\begin{equation}
    \bar{\mathbf{h}} \sim \mathcal{CN}(\mathbf{0}, \mathbf{R}),
\end{equation}
where $\mathbf{R} = \mathbb{E}[\bar{\mathbf{h}}\bar{\mathbf{h}}^H]$ is the full spatial-frequency covariance matrix. Crucially, $\mathbf{R}$ encapsulates the essential geometric features of the propagation environment---such as angles of arrival and delay clusters. This geometric structure defines the optimal KLT basis $\mathbf{U} \in \mathbb{C}^{N \times N}$, derived from the eigendecomposition $\mathbf{R} = \mathbf{U} \mathbf{\Lambda} \mathbf{U}^{H}$, where $\mathbf{\Lambda}$ contains the ordered eigenvalues. In practice, the UE obtains this eigenbasis via long-term covariance estimation. For analytical tractability, we treat the resulting basis as an accurate reference KLT basis $\mathbf{U}$. This assumption is justified when the propagation geometry remains locally stationary, so that the channel statistics vary slowly and a single covariance matrix can be reliably estimated from sufficient samples. When the geometry varies across users or time, the channel distribution becomes multi-modal and a Gaussian-mixture model is more appropriate; such an extension is explored in \cite{park2026fundamentallimitscsicompression}. Our focus is on the fundamental impact of feedback compression under basis mismatch, decoupling it from covariance estimation errors.

\subsection{CSI Compression Problem} 
The objective is to encode $\bar{\mathbf{h}}$ into a finite bitstream and reconstruct it at the BS. This forms a source coding problem defined by an encoder-decoder pair $(f, g)$. The encoder maps $\bar{\mathbf{h}}$ to an index $s$, and the decoder reconstructs $\hat{\bar{\mathbf{h}}}=g(s)$, aiming to minimize the MSE distortion $D=\frac{1}{N}\mathbb{E}[\|\bar{\mathbf{h}}-\hat{\bar{\mathbf{h}}}\|^2]$.

Ideally, TC achieves the R-D limit by utilizing the optimal KLT basis $\mathbf{U}$ at both ends. However, in FDD systems, the BS cannot directly access the UE-side eigenbasis $\mathbf{U}$ because the uplink and downlink operate at different carrier frequencies, making reciprocity-based inference unavailable. Consequently, the basis itself must be explicitly quantized into $\hat{\mathbf{U}}$ and fed back. This introduces an unavoidable basis mismatch in the transform chain: the encoder decorrelates using the perfect basis $\mathbf{U}$, whereas the decoder reconstructs using the quantized basis $\hat{\mathbf{U}} \neq \mathbf{U}$. This mismatch projects quantization noise onto a misaligned subspace, creating a unique distortion component that we analyze and optimize in this paper.

\section{CSI Compression using TC \\under Basis Mismatch}

We formulate a practical TC framework tailored for FDD massive MIMO systems. The proposed architecture follows the standard transform coding paradigm---linear decorrelation, bit allocation, scalar quantization, and entropy coding---but incorporates a split-structure basis feedback to address the massive dimension of the channel.

\subsection{Separable Covariance Structure}
To manage the prohibitive dimensionality of the full covariance matrix $\mathbf{R} \in \mathbb{C}^{N \times N}$ ($N=N_c N_t$), we exploit the separable spatial--frequency structure of massive MIMO channels. Assuming independence between spatial and frequency-domain correlations, the covariance matrix can be approximated as $\mathbf{R} \approx \mathbf{R}_s \otimes \mathbf{R}_f$, where $\mathbf{R}_s \in \mathbb{C}^{N_t \times N_t}$ captures the spatial correlation across antennas and $\mathbf{R}_f \in \mathbb{C}^{N_c \times N_c}$ captures the frequency-domain correlation across subcarriers. This yields the basis decomposition, $\mathbf{U} \approx \mathbf{U}_s \otimes \mathbf{U}_f$, where $\mathbf{U}_s,\mathbf{U}_f$ denote the basis of $\mathbf{R}_s,\mathbf{R}_f$, respectively.

\subsection{Long-term CSI: Basis Quantization}
Based on the separable structure, the UE independently quantizes the dominant columns of $\mathbf{U}_s$ and $\mathbf{U}_f$ using Grassmannian codebooks on $G(\cdot,1)$ and applies a re-orthonormalization step to construct unitary matrices $\hat{\mathbf{U}}_s$ and $\hat{\mathbf{U}}_f$. The BS reconstructs the effective basis as $\hat{\mathbf{U}} = \hat{\mathbf{U}}_s \otimes \hat{\mathbf{U}}_f$. While the Kronecker structure reduces the feedback complexity, characterizing the exact distortion induced by column-wise quantization and subsequent re-orthonormalization is analytically intractable. Therefore, for R-D analysis, we adopt a random vector quantization (RVQ) and treat each dominant column of the effective basis as independently quantized on $G(N,1)$. The re-orthonormalization step is used only to enforce unitarity and is not explicitly modeled in the distortion analysis.

\subsection{Short-term CSI: Channel Coefficient Compression}
For each channel realization $\bar{\mathbf{h}}$, the feedback processing follows four stages. First, decorrelation is performed by projecting $\bar{\mathbf{h}}$ onto the true eigenbasis $\mathbf{U}$ to obtain $\tilde{\mathbf{h}} = \mathbf{U}^{H} \bar{\mathbf{h}}$. Second, bit allocation is determined via RWF. Third, active coefficients are quantized using an entropy constrained scalar quantizer (ECSQ), yielding indices $q_m$ and reconstructed values $\hat{\tilde{\mathbf{h}}}$. Finally, entropy coding compresses the indices to meet the target rate.

\subsection{Reconstruction with Mismatch}
The BS decodes the entropy-coded stream to retrieve $\hat{\tilde{\mathbf{h}}}$ and reconstructs the channel using the quantized basis $\hat{\mathbf{U}}$:
\begin{equation} \label{eq:mismatch_recon}
    \hat{\bar{\mathbf{h}}} = \hat{\mathbf{U}} \hat{\tilde{\mathbf{h}}}.
\end{equation}
The critical observation is the mismatch in the transform chain: the decorrelation (Step 1) uses $\mathbf{U}$, while the reconstruction \eqref{eq:mismatch_recon} uses $\hat{\mathbf{U}}$. This implies that the quantization noise in $\hat{\tilde{\mathbf{h}}}$ is projected onto a misaligned subspace, introducing additional distortion beyond standard quantization error.

\section{E2E MSE Analysis}\label{sec:e2e_mse}
We characterize the E2E MSE of the proposed feedback architecture when the decoder reconstructs CSI using a \emph{quantized} eigenbasis. The goal is to express the distortion as a function of the coefficient rate \(R_q\) and the basis rate \(R_0\), which enables joint rate allocation in the sequel.

\subsection{Rate split and distortion decomposition}\label{subsec:rate_split_decomp}
Let \(N\) denote the channel dimension per realization (e.g., \(N=N_cN_t\)), and let \(R_{\mathrm{total}}\) be the feedback budget in \emph{bits per complex dimension per realization}. We split the budget into $R_{\mathrm{total}} = R_0 + R_q$, where \(R_0\) is the \emph{amortized} basis rate used to convey a quantized eigenbasis. Assume the channel covariance eigenbasis \(\mathbf{U}\) is constant over a coherence block of \(\tau\) channel realizations. If the UE quantizes only the \(p\) dominant eigenvectors using \(B_{\mathrm{basis}}\) bits per block, then
\begin{equation}\label{eq:R0_def}
R_0 \triangleq \frac{B_{\mathrm{basis}}}{N\tau}.
\end{equation}
Equivalently, the basis bits per dominant column are \(B = B_{\mathrm{basis}}/p = (R_0N\tau)/p\). In addition, \(R_q\) is the coefficient rate used to quantize the instantaneous transform coefficients in each realization.

Let \(\bar{\mathbf{h}}\in\mathbb{C}^N\) denote the source vector to be fed back, with KLT coefficients \(\tilde{\mathbf{h}}=\mathbf{U}^{H}\bar{\mathbf{h}}\). The encoder sends a quantized coefficient vector \(\hat{\tilde{\mathbf{h}}}\), while the decoder reconstructs using a quantized basis \(\hat{\mathbf{U}}\), i.e., $\hat{\bar{\mathbf{h}}}=\hat{\mathbf{U}}\hat{\tilde{\mathbf{h}}}.$ The E2E distortion (MSE per complex dimension) is
\begin{equation}
D_{\mathrm{E2E}} \triangleq \frac{1}{N}\mathbb{E}\left[\|\bar{\mathbf{h}}-\hat{\bar{\mathbf{h}}}\|_2^2\right].
\end{equation}
Using \(\bar{\mathbf{h}}=\mathbf{U}\tilde{\mathbf{h}}\) and adding/subtracting \(\mathbf{U}\hat{\tilde{\mathbf{h}}}\), we obtain the exact decomposition
\begin{align}\label{eq:decomp_exact}
D_{\mathrm{E2E}}
&= \frac{1}{N}\mathbb{E}\left[\big\|\mathbf{U}(\tilde{\mathbf{h}}-\hat{\tilde{\mathbf{h}}})
+(\mathbf{U}-\hat{\mathbf{U}})\hat{\tilde{\mathbf{h}}}\big\|_2^2\right]\notag\\
&= \underbrace{\frac{1}{N}\mathbb{E}\left[\|\tilde{\mathbf{h}}-\hat{\tilde{\mathbf{h}}}\|_2^2\right]}_{T_1}
+\underbrace{\frac{1}{N}\mathbb{E}\left[\|(\mathbf{U}-\hat{\mathbf{U}})\hat{\tilde{\mathbf{h}}}\|_2^2\right]}_{T_2}
\nonumber\\
&\quad+\underbrace{\frac{2}{N}\Re\mathbb{E}\left[(\tilde{\mathbf{h}}-\hat{\tilde{\mathbf{h}}})^{H}\mathbf{U}^{H}(\mathbf{U}-\hat{\mathbf{U}})\hat{\tilde{\mathbf{h}}}\right]}_{T_3}.
\end{align}
Here \(T_1\) and \(T_2\) represent the coefficient and basis mismatch distortions, which we denote by \(D_q\) and \(D_0\), respectively, while \(T_3\) is a cross-term.

\subsection{Coefficient and basis distortion models}\label{subsec:component_models}
Let \(\mathbf{C}_{\tilde{\mathbf{h}}} = \operatorname{diag}(\lambda_1, \dots, \lambda_N)\) be the covariance of \(\tilde{\mathbf{h}}\) in the true KLT domain (w.l.o.g.\ after ordering eigenvalues). Under the Gaussian test channel (RWF), the coefficient distortion at rate \(R_q\) is given by the standard Gaussian R-D expression \cite{cover2006elements}.

\begin{lemma}[Coefficient distortion]\label{lem:coef_dist}
For coefficient rate \(R_q\) (bits/complex-dimension), the transform-domain distortion is
\begin{equation}\label{eq:Dq_def}
D_q(R_q)=\frac{1}{N}\sum_{m=1}^N \min\!\big(\lambda_m,\mu(R_q)\big),
\end{equation}
where the water level \(\mu(R_q)\) is chosen such that
\begin{equation}\label{eq:waterlevel}
\sum_{m=1}^N \max\!\left(0,\log_2\frac{\lambda_m}{\mu(R_q)}\right)=NR_q.
\end{equation}
\end{lemma}

For the basis error, we model quantization of the \(p\) dominant eigenvectors via independent RVQ on the Grassmann manifold.

\begin{lemma}[RVQ basis error scaling]\label{lem:rvq_scaling}
Let \(\mathbf{u}\in\mathbb{C}^N\) be a unit-norm vector and \(\hat{\mathbf{u}}\) its RVQ quantization using \(B\) bits on \(G(N,1)\). Then the expected squared chordal distortion satisfies \cite{Jindal2006,love2003Grassmanian}
\begin{equation}\label{eq:rvq_chordal}
\mathbb{E}\left[d_c^2(\mathbf{u},\hat{\mathbf{u}})\right]\doteq c_N\,2^{-B/(N-1)},
\end{equation}
where \(c_N\) is a dimension-dependent constant (often treated as \(\Theta(1)\) for fixed \(N\)).
\end{lemma}

\subsection{Main result}\label{subsec:main_theorem}
We now state and prove a closed-form E2E distortion characterization that explicitly reveals the trade-off between allocating bits to coefficients versus to the basis.

\begin{theorem}[E2E distortion under basis mismatch]\label{thm:e2e_basis_mismatch}
Assume: (i) coefficient quantization uses subtractive dithering (or an equivalent high-rate model) so that the quantization error is conditionally unbiased and uncorrelated with the input \cite{Gray1998quantize}; (ii) basis quantization is performed once per coherence block and is independent of the instantaneous coefficient quantization; and (iii) the \(p\) dominant eigenvectors are quantized independently via RVQ with \(B=(R_0N\tau)/p\) bits per column. Then the E2E distortion satisfies
\begin{equation}\label{eq:thm_e2e}
D_{\mathrm{E2E}}(R_q,R_0) = D_q(R_q) + D_0(R_0) + o(1),
\end{equation}
where \(D_q(R_q)\) is given in \eqref{eq:Dq_def}, \(o(1)\) captures the vanishing cross-term in the high-rate regime, and the basis mismatch term admits the approximation
\begin{equation}\label{eq:D0_final}
D_0(R_0)\doteq \alpha_0\!\left(\sum_{m=1}^p \lambda_m\right)2^{-\beta_0 R_0},\quad \beta_0=\frac{N\tau}{p(N-1)},
\end{equation}
with \(\alpha_0 = \frac{c_N}{N}\) (and in particular \(\alpha_0=\frac{N-1}{N^2}\) under the standard normalization used in \cite{Jindal2006}).
\end{theorem}

\begin{proof}
Start from the exact decomposition \eqref{eq:decomp_exact}. Since \(\mathbf{U}\) is unitary, \(\|\mathbf{U}\mathbf{x}\|_2=\|\mathbf{x}\|_2\) and thus
\[
T_1=\frac{1}{N}\mathbb{E}\left[\|\tilde{\mathbf{h}}-\hat{\tilde{\mathbf{h}}}\|_2^2\right]=D_q(R_q),
\]
where \(D_q(R_q)\) is given by Lemma~\ref{lem:coef_dist}.

Let \(\mathbf{e}_q\triangleq \tilde{\mathbf{h}}-\hat{\tilde{\mathbf{h}}}\) and \(\mathbf{M}\triangleq \mathbf{U}^{H}(\mathbf{U}-\hat{\mathbf{U}})\). Then
\[
T_3=\frac{2}{N}\Re\,\mathbb{E}\left[\mathbf{e}_q^{H}\mathbf{M}\hat{\tilde{\mathbf{h}}}\right].
\]
Under subtractive dithering (or a high-rate model), \(\mathbb{E}[\mathbf{e}_q\mid\tilde{\mathbf{h}}]=\mathbf{0}\) and \(\mathbf{e}_q\) is uncorrelated with \(\tilde{\mathbf{h}}\) \cite{Gray1998quantize}. Moreover, by time-scale separation, the basis error \(\mathbf{M}\) is a function of the (long-term) basis quantization and is independent of the instantaneous coefficient quantization, hence independent of \(\mathbf{e}_q\). Therefore,
\[
\mathbb{E}\left[\mathbf{e}_q^{H}\mathbf{M}\hat{\tilde{\mathbf{h}}}\right]
=\mathbb{E}\left[\mathbb{E}[\mathbf{e}_q^{H}\mid \tilde{\mathbf{h}},\mathbf{M}]\;\mathbf{M}\hat{\tilde{\mathbf{h}}}\right]
=\mathbb{E}\left[\mathbf{0}^{H}\mathbf{M}\hat{\tilde{\mathbf{h}}}\right]=0,
\]
which implies \(T_3=o(1)\) (and is exactly zero under the ideal dithered model).

We write $T_2 = \frac{1}{N}\mathbb{E}[\hat{\tilde{\mathbf{h}}}^{H}(\mathbf{U}-\hat{\mathbf{U}})^{H}(\mathbf{U}-\hat{\mathbf{U}})\hat{\tilde{\mathbf{h}}}]$. Assume only the first \(p\) eigenvectors are quantized (the remaining columns are either not used or treated as perfectly known). Then \(\mathbf{U}-\hat{\mathbf{U}}\) has nonzero columns only in \(\{1,\ldots,p\}\), and expanding the quadratic form yields
\[
T_2=\frac{1}{N}\sum_{m=1}^p \mathbb{E}\left[|\hat{\tilde h}_m|^2\,\|\mathbf{u}_m-\hat{\mathbf{u}}_m\|_2^2\right] +\text{(cross-column terms)}.
\]
Under independent RVQ across columns and isotropy, the cross-column terms vanish in expectation (or are negligible compared to the dominant diagonal terms), giving
\[
T_2 \doteq \frac{1}{N}\sum_{m=1}^p \mathbb{E}\left[|\hat{\tilde h}_m|^2\right]\mathbb{E}\left[\|\mathbf{u}_m-\hat{\mathbf{u}}_m\|_2^2\right].
\]
Moreover, \(\mathbb{E}[|\hat{\tilde h}_m|^2]\le \mathbb{E}[|\tilde h_m|^2]=\lambda_m\), so
\[
T_2 \le \frac{1}{N}\left(\sum_{m=1}^p \lambda_m\right)\mathbb{E}\left[\|\mathbf{u}-\hat{\mathbf{u}}\|_2^2\right],
\]
where \((\mathbf{u},\hat{\mathbf{u}})\) denotes a generic dominant eigenvector and its RVQ quantization. For small quantization error, \(\|\mathbf{u}-\hat{\mathbf{u}}\|_2^2 \doteq d_c^2(\mathbf{u},\hat{\mathbf{u}})\) (up to a dimension-dependent constant due to phase alignment on \(G(N,1)\)), hence by Lemma~\ref{lem:rvq_scaling},
\[
\mathbb{E}\left[\|\mathbf{u}-\hat{\mathbf{u}}\|_2^2\right]\doteq c_N\,2^{-B/(N-1)}.
\]
Substituting \(B=(R_0N\tau)/p\) yields
\[
T_2 \doteq \frac{c_N}{N}\left(\sum_{m=1}^p \lambda_m\right)2^{-\frac{R_0N\tau}{p(N-1)}} = \alpha_0\left(\sum_{m=1}^p \lambda_m\right)2^{-\beta_0R_0},
\]
which is exactly \eqref{eq:D0_final}. This completes the proof.
\end{proof}

Theorem~\ref{thm:e2e_basis_mismatch} cleanly separates short-term and long-term errors: \(D_q(R_q)\) is governed by reverse water-filling over eigenmodes, while \(D_0(R_0)\) decays exponentially with the amortized basis rate. This structure makes the optimal rate split \((R_q,R_0)\) analytically tractable. The closed-form expression relies on the RVQ model 
and the high-rate dithering assumption; extending the 
characterization to low-rate regimes and structured 
codebooks is left for future work.

\section{Optimal Rate Allocation\\ and Phase Transition Analysis}

Having characterized the E2E distortion $D_{E2E}(R_q, R_0)$ in Theorem 1, we now determine the optimal split of the total budget $R_{total}$ between the coefficient rate $R_q$ and the amortized basis rate $R_0$. The optimization problem is formulated as minimizing $D_{E2E} = D_q(R_q) + D_0(R_0)$ subject to $R_q + R_0 \le R_{total}$. Since both distortion terms are non-increasing in their respective rates and are piecewise smooth under a fixed active set of RWF, the problem admits a well-defined optimum.

\begin{theorem}[Optimal Rate Split Condition]
For a sufficiently large total rate $R_{total}$, the optimal rate allocation $(R_q^*, R_0^*)$ satisfies the condition where the marginal reduction in coefficient distortion equals the marginal reduction in basis mismatch distortion:
\begin{equation} \label{eq:opt_condition}
    \mu(R_q^*) = \beta_0 D_0(R_0^*),
\end{equation}
where $\beta_0 = \frac{N \tau}{p(N-1)}$ is the basis rate scaling factor defined in Theorem 1.
\end{theorem}

\begin{proof}
The optimal solution must satisfy the condition that the magnitude of the gradients with respect to the rates are equal: $|\frac{\partial D_q}{\partial R_q}| = |\frac{\partial D_0}{\partial R_0}|$. For a fixed active set of RWF, standard high-rate arguments give the local slope $\frac{\partial D_q}{\partial R_q} = -(\ln 2)\mu$. Similarly, differentiating the exponential basis distortion term $D_0(R_0) \propto 2^{-\beta_0 R_0}$ yields $\frac{\partial D_0}{\partial R_0} = -(\ln 2)\beta_0 D_0(R_0)$, which completes the proof.
\end{proof}

Theorem 2 characterizes the optimal operating point by equating the marginal distortion reductions achieved by allocating an additional bit to coefficient quantization and to basis feedback. In other words, the optimal split satisfies $\mu(R_q^*)=\beta_0 D_0(R_0^*)$, which directly leads to a phase-transition behavior in the rate allocation.

\begin{theorem}[Phase Transition Threshold]
There exists a critical rate threshold $R_{th}$ that governs the feedback strategy:
\begin{itemize}
    \item \textbf{Inactive Regime ($R_{total} \le R_{th}$):} It is optimal to allocate all bits to coefficients ($R_0^* = 0, R_q^* = R_{total}$).
    \item \textbf{Active Regime ($R_{total} > R_{th}$):} It is optimal to update the basis ($R_0^* > 0$), with rates satisfying \eqref{eq:opt_condition}.
\end{itemize}
The threshold is derived by evaluating the optimality condition at the boundary $R_0=0$:
\begin{equation} \label{eq:threshold}
    R_{th} \approx \frac{K}{N} \log_2 \left( \frac{(\prod_{m \in \mathcal{S}} \lambda_m)^{1/K}}{\beta_0 D_0(0)} \right),
\end{equation}
where $\mathcal{S}$ is the active set at rate $R_{th}$, and $K=|\mathcal{S}|$ is the number of active eigenmodes.
\end{theorem}

\begin{proof}
Whether basis feedback is active is decided by the one-sided derivative at $R_0=0$: if allocating an infinitesimal basis rate decreases distortion, the optimum moves to $R_0^*>0$; otherwise $R_0^*=0$ remains optimal. The boundary condition is
$\mu(R_{\mathrm{th}})=\beta_0D_0(0).$
For a fixed active set $\mathcal{S}$ with $K=|\mathcal{S}|$, reverse water-filling gives
\[
\mu(R_q)=\left(\prod_{m\in\mathcal{S}}\lambda_m\right)^{1/K}2^{-NR_q/K},
\]
and substituting $R_q=R_{\mathrm{th}}$ yields \eqref{eq:threshold}.
\end{proof}

Intuitively, $R_{th}$ marks the point where the marginal gain from basis refinement exceeds that of coefficient quantization. Below $R_{th}$, reducing coefficient distortion is more effective, making basis feedback unnecessary, whereas above $R_{th}$, the basis-mismatch term $D_0(R_0)$ becomes the dominant bottleneck.

\begin{figure*}[t]
    \centering
    \begin{minipage}[t]{0.45\linewidth} 
        \centering
        \includegraphics[height=6cm, keepaspectratio]{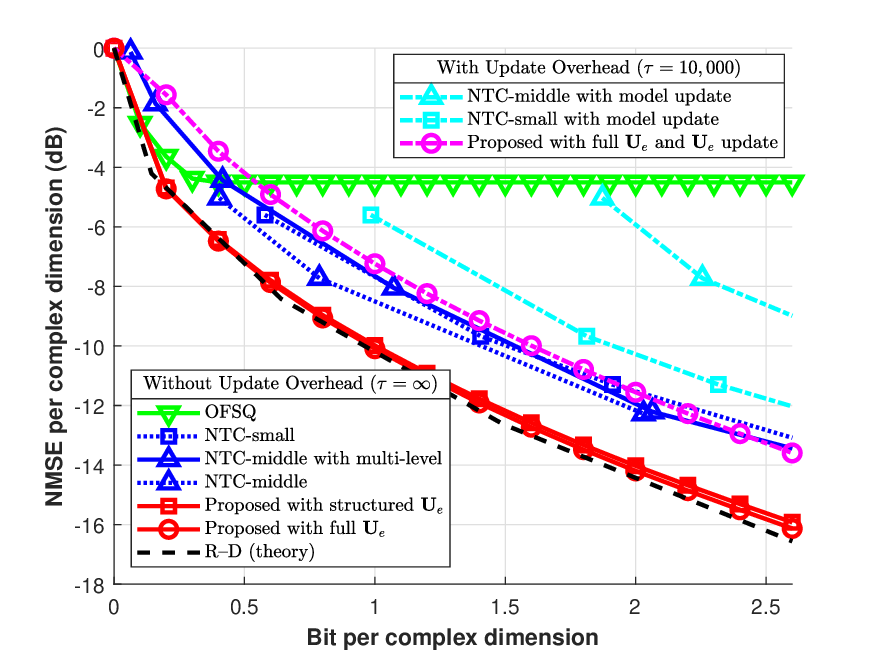}
        \caption{R-D performance comparison on the correlated Gaussian channel ($N_c=N_t=32, \rho=0.8$).}
        \label{Rayleigh_RD_result}
    \end{minipage}
    \hfill
    \begin{minipage}[t]{0.45\linewidth}
        \centering
        \includegraphics[height=6cm, keepaspectratio]{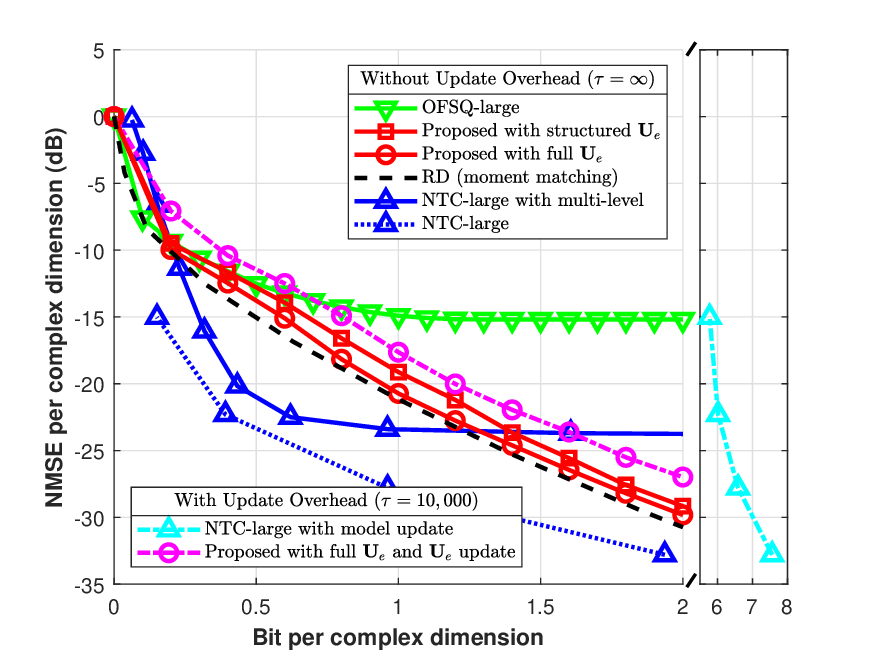}
        \caption{R-D performance comparison on the realistic COST2100 indoor \SI{5.3}{\giga\hertz} channel.}
        \label{COST2100_RD_result}
    \end{minipage}
\end{figure*}

\section{Numerical Results}
We evaluate the proposed CSI compression framework against state-of-the-art methods in a single-cell FDD massive MIMO system with $N_t=32$ antennas and $N_c=32$ subcarriers. Feedback rates are normalized by the dimension $N=N_c N_t$.

\textbf{Channel Models and Baselines:} We consider two channel environments. For theoretical validation, we generate a correlated Gaussian channel with Kronecker exponential correlation ($[\mathbf{R}_s]_{i,j}=[\mathbf{R}_f]_{i,j}=\rho^{|i-j|}$, $\rho=0.8$). For realistic evaluation, we use the COST2100 indoor \SI{5.3}{\giga\hertz} scenario~\cite{COST2100}, which yields non-Gaussian channels due to geometry-based multipath propagation. Since the exact R-D function is intractable in this case, we adopt a Gaussian moment-matched benchmark by estimating the sample covariance $\hat{\mathbf{R}}$ and using the R-D function of $\mathcal{CN}(\mathbf{0},\hat{\mathbf{R}})$ as an upper bound. Baselines include the Gaussian R-D bound, OFSQ~\cite{Liotopoulos2025OFSQ}, a learning-based scalar quantization method with ordered dropout for rate adaptation, and state-of-the-art NTC autoencoders of varying sizes, including the rate-adaptive NTC-multi-level variant~\cite{park2024multirate}.

\textbf{Implementation \& Dynamics:} For the proposed TC, we use bisection search for the RWF water-level and Lloyd--Max optimized ECSQ. We evaluate both a static regime ($\tau\to\infty$) and a dynamic regime ($\tau=10{,}000$), where the latter accounts for update overhead via $R_{\text{eff}}=R_{\text{total}}-B_{\text{update}}/N\tau$.

\textbf{Rate-Distortion Performance:} Fig.~\ref{Rayleigh_RD_result} confirms that on the correlated Gaussian channel, the proposed TC with full $\mathbf{U}$ tightly matches the theoretical bound, and the Structured-$\mathbf{U}$ variant shows negligible degradation, implying that the separable approximation captures the dominant eigenstructure effectively. On the realistic COST2100 channel (Fig.~\ref{COST2100_RD_result}), OFSQ and NTC-multi-level saturate in the high-rate regime ($>0.6$ bits/dim). In contrast, our framework continuously reduces distortion by optimally balancing basis and coefficient rates, approaching the performance of the computationally intensive NTC-large without requiring training.

\textbf{Dynamics and Complexity:} In dynamic scenarios, update overhead is critical. Adapting DL models requires retransmitting full FP32 weights ($B_{\text{update}} \approx 32\times\text{\#params}$), causing the effective rate to collapse as shown in Fig.~\ref{COST2100_RD_result}. Conversely, our framework only refreshes the lightweight quantized basis indices ($R_0$), maintaining robust performance even at $\tau=10{,}000$. In terms of complexity, the separable TC scheme requires only $2{,}048$ parameters and approximately $1.31\times 10^5$ FLOPs per block, whereas the state-of-the-art NTC-large demands over $2.1$ million parameters and $\approx 8.0\times 10^8$ FLOPs---corresponding to an $878\times$ reduction in parameters and a $6{,}000\times$ reduction in FLOPs, making the proposed scheme uniquely feasible for resource-constrained UEs.

\textbf{Discussion:} The proposed framework is expected to perform well in practical channels whose statistics are well captured by second-order structure, as is often the case in rich-scattering propagation environments. In such scenarios, the Gaussian-based analysis provides an accurate approximation, and the resulting rate-split rule remains effective.

\section{Conclusion}
This paper investigated FDD massive MIMO CSI feedback via a TC R-D analysis, modeling basis quantization with an RVQ surrogate for Grassmannian quantization. We derived a closed-form E2E distortion expression and an optimal rate-splitting rule with a phase-transition threshold for activating basis feedback. Simulations confirm near-optimal performance, robustness to update overhead, and orders-of-magnitude complexity reduction over DL methods.

\section*{ACKNOWLEDGMENT}
This work was supported by the National Research Foundation of Korea (NRF) Grant funded by the Korea government (MSIT) under Grant 2022R1A5A1027646, and by the Korea Research Institute for defense Technology planning and advancement (KRIT) - grant funded by the Defense Acquisition Program Administration (DAPA) (KRIT-CT-22-078).
\newpage
\bibliographystyle{IEEEtran}
\bibliography{bibliography}

\end{document}